\documentclass[journal]{IEEEtran}
\makeatletter
\usepackage{caption}
\usepackage{subcaption}
\usepackage{textcomp}
\usepackage{fixltx2e}
\usepackage{stmaryrd}
\usepackage{multirow}
\usepackage{array}
\usepackage{float}
\usepackage{graphicx}
\usepackage[cmex10]{amsmath}
\usepackage{cite}
\usepackage{amsthm}
\usepackage{amsfonts}
\usepackage{amssymb}
\usepackage{textcomp}
\IEEEoverridecommandlockouts
\usepackage{algorithm}
\usepackage{balance}
\usepackage{algpseudocode}
\newtheorem{theorem}{Theorem}[section]

\usepackage{adjustbox}
\newcolumntype{R}[2]{%
    >{\adjustbox{angle=#1,lap=\width-(#2)}\bgroup}%
    l%
    <{\egroup}%
}

\DeclareCaptionFormat{algorithm}{\vspace{0ex}{%
 \parbox[c][0.9em][c]{\textwidth}{#1#2#3}}\par\rule{\linewidth}{-0.5pt}}
\begin{document}
\addtolength{\tabcolsep}{-5pt}
\DeclareGraphicsExtensions{.pdf,.png,.gif,.jpg}

\title{Graph Fourier Transform based on Directed Laplacian}

\author{Rahul~Singh,
       Abhishek~Chakraborty,~\IEEEmembership{Graduate Student Member,~IEEE,}
       and~B.~S.~Manoj,~\IEEEmembership{Senior Member,~IEEE}
\thanks{R. Singh, A. Chakraborty, and B. S. Manoj are with the Department
of Avionics, Indian Institute of Space Science and Technology, Thiruvananthapuram, Kerala 695547 India (e-mail: rahul.s.in@ieee.org, abhishek2003slg@ieee.org, bsmanoj@ieee.org).}
}

\maketitle

\begin{abstract}
\boldmath
In this paper, we redefine the Graph Fourier Transform (GFT) under the DSP$_\mathrm{G}$ framework. We consider the Jordan eigenvectors of the directed Laplacian as graph harmonics and the corresponding eigenvalues as the graph frequencies. For this purpose, we propose a shift operator based on the directed Laplacian of a graph. Based on our shift operator, we then define total variation of graph signals, which is used in frequency ordering. We achieve natural frequency ordering and interpretation via the proposed definition of GFT. Moreover, we show that our proposed shift operator makes the LSI filters under DSP$_\mathrm{G}$ to become polynomial in the directed Laplacian.
\end{abstract}


\section{Introduction}
\label{sec:Introduction}
\IEEEPARstart
{D}{ata} defined on network-like structures are encountered in a large number of scenarios including molecular interaction in biological systems, computer networks, sensor networks, social and citation networks, the Internet and the World Wide Web, power grids, transportation networks~\cite{Newmanbook2010,Watts2004}, and many more. Such data can be visualized as a set of scalar values, known as a graph signal, lying on a particular structure, i.e., a graph. In computer graphics, data defined on any geometrical shape described by polygon meshes can be formulated as a graph signal~\cite{graphics1996}. \\
\indent The irregular structure of the underlying graph, as opposed to the regular structure in case of time-series and image signals dealt in classical signal processing~\cite{DSPOpp,Mallat}, imposes a great challenge in analysis and processing of graph signals. Fortunately, recent work toward the development of important concepts and tools, extending classical signal processing theory, including sampling and interpolation on graphs~\cite{Chen2015,Gadde2013,Gadde2014}, graph-based transforms~\cite{Wavelet2011,Shuman2013,Shuman2015,Moura2014,MouraBig2014}, and graph filters~\cite{Moura2013,IIR2015} have enriched the field of graph signal processing. These tools have been utilized in solving a variety of problems such as signal recovery on graphs~\cite{Ribiero2015,Recovery2015,Reconstruction2015}, clustering and community detection~\cite{Tremblay2014,Dong2014}, graph signal denoising~\cite{Denoising2014}, and semi-supervised classification on graphs\cite{Semi2013}. \\
\indent Transforms aimed at frequency analysis of graph signals facilitate efficient handling of the data and remain at the heart of graph signal processing. In literature, there exist two frameworks for frequency analysis and processing of graph signals: (i) Laplacian matrix based approach, and (ii) weight matrix based approach. The existing Laplacian based approach~\cite{Wavelet2011,Shuman2013} is limited to the analysis of graph signals lying on undirected graphs with real non-negative weights. In the Laplacian based approach, eigendecomposition of the graph Laplacian $\mathbf{L}$ is used for frequency analysis of graph signals. Mathematically, Graph Fourier Transform (GFT) of a graph signal $\mathbf{f}$ is defined as $\mathbf{\hat{f}} = \mathbf{U}^T \mathbf{f}$, where $\mathbf{U}~=~[\mathbf{u_0}~\mathbf{u_1}~\ldots~\mathbf{u_{N-1}}]$ is the matrix in which columns are the eigenvectors of $\mathbf{L}$. Here, the frequency ordering is based on the quadratic form and turns out to be ``natural''. Natural frequency ordering means that the small eigenvalues correspond to low frequencies and vice-versa. On the other hand, the weight matrix based framework~\cite{Moura2013,Moura2014}, also referred as the Discrete Signal Processing on Graphs (DSP$_{\mathrm{G}}$) framework, has been built on the graph shift operator. The weight matrix~$\mathrm{W}$ of the graph acts as the shift operator and the shifted version of a graph signal~$\mathbf{f}$ can be found as $\mathbf{\tilde{f}} = \mathrm{W}\mathbf{f}$. Shift operator is treated as the elementary Linear Shift Invariant (LSI) graph filter, which essentially is a polynomial in~$\mathrm{W}$. Based on the shift operator, authors define total variation (TV) on graphs that they utilize for ordering of graph frequencies. In DSP$_{\mathrm{G}}$, the graph Fourier transform of a graph signal $\mathbf{f}$ is defined as~$\mathbf{\hat{f}} = \mathbf{V}^{-1} \mathbf{f}$, where $\mathbf{V}$ is the matrix with the Jordan eigenvectors of $\mathrm{W}$ as its columns ($\mathrm{W}$ is decomposed in Jordan canonical form as $\mathrm{W} = \mathbf{VJV}^{-1}$). The Jordan eigenvectors of $\mathrm{W}$ are used as graph Fourier basis and the eigenvalues of~$\mathrm{W}$ act as graph frequencies. \\
\indent Although DSP$_{\mathrm{G}}$ is applicable to directed graphs with negative or complex edge weights as well, it does not provide ``natural'' intuition of frequency. The eigenvalue of $\mathrm{W}$ with \textit{maximum} absolute value acts as the \textit{lowest} frequency and as one moves away from this eigenvalue in the complex frequency plane, the frequency increases. Thus, frequency ordering is not natural and there is an overhead of frequency ordering as well. Also, the interpretation of frequency is not intuitive --- for example, in general, a constant graph signal results in low as well as high frequency components in the spectral domain. \\
\indent In this paper, we redefine Graph Fourier Transform (GFT) under DSP$_\mathrm{G}$. In the new definition of GFT, the Jordan eigenvectors of the directed Laplacian matrix are treated as the graph Fourier basis and the corresponding eigenvalues constitute the graph spectrum. The directed Laplacian matrix of a graph is a simple extension of the symmetric Laplacian discussed in~\cite{Shuman2013} to directed graphs. To redefine GFT under DSP$_\mathrm{G}$, we first propose a shift operator derived from the directed Laplacian. Then, we utilize this shift operator to define total variation (TV) of a graph signal which is used for frequency ordering. We observe ``natural'' frequency ordering as well as better intuition as compared to the existing GFT definition under DSP$_{\mathrm{G}}$ approach. Moreover, the new definition of GFT links the DSP$_\mathrm{G}$ framework to the existing Laplacian based approach. We also show that considering our proposed shift operator, the LSI filters under DSP$_\mathrm{G}$ become polynomial in the directed Laplacian. \\
\indent Rest of this paper is organized as follows. In Section~\ref{sec:FreqAnalysis}, we discuss the directed Laplacian of a graph followed by the proposed shift operator and total variation on graph. Then, we redefine the Graph Fourier Transform based on directed Laplacian. We describe LSI graph filters in Section~\ref{sec:Filtering} and then conclude the paper in Section~\ref{sec:Conclusion}. 
\section{Frequency Analysis of Graph Signals}
\label{sec:FreqAnalysis}
First, we present directed Laplacian matrix of a graph and then, derive the shift operator from it. Next, we define total variation of a graph signal which is utilized in frequency ordering. Finally, we redefine Graph Fourier Transform under the DSP$_\mathrm{G}$ framework.
\subsection{Graph Signals}
\label{subsec:GraphSignals}
A graph signal is a collection of values defined on a complex and irregular structure modeled as a graph. A graph is represented as ${\cal{G}} = ({\cal{V}} , \mathrm{W})$, where ${\cal{V}} = \{v_0, v_1,\ldots, v_{N-1}\}$ is the set of vertices (or nodes) and $\mathrm{W}$ is the weight matrix of the graph in which an element $w_{ij}$ represents the weight of the directed edge from node $j$ to node $i$. Moreover, a graph signal is represented as an $N$-dimensional vector $\mathbf{f} = [f (1),f (2),\ldots, f (N)] ^T \in \mathbb{C}^{N}$, where $f(i)$ is the value of the graph signal at node $i$ and $N = |\cal{V}|$ is the total number of nodes in the graph.
\subsection{Directed Laplacian}
\label{subsec:dirLaplacian}
As discussed in~\cite{Shuman2013}, the graph Laplacian for undirected graphs is a symmetric difference operator $\mathbf{L} = \mathrm{D - W}$, where $\mathrm{D}$ is the degree matrix of the graph and $\mathrm{W}$ is the weight matrix of the graph. However, in case of directed graphs (or digraphs), the weight matrix $\mathrm{W}$ of a graph is not symmetric. In addition, the degree of a vertex can be defined in two ways~\cite{Newmanbook2010} --- in-degree and out-degree. In-degree of a node~$i$ is estimated as $d_i^{in} = \sum_{j=1}^N w_{ij}$, whereas, out-degree of the node~$i$ can be calculated as $d_i^{out} = \sum_{j=1}^N w_{ji}$. We consider in-degree matrix and define the directed Laplacian $\mathbf{L}$ of a graph as
\begin{equation}
\label{eq:outLap}
\mathbf{L} = \mathrm{D_{in} - W},
\end{equation}
\noindent where $\mathrm{D_{in}} = \mathrm{diag}\left(\{d_i^{in}\}_{i=1}^{N}\right)$ is the in-degree matrix. Fig.~\ref{fig:dirGraphExample}(a) shows an example weighted directed graph and the corresponding matrices are shown in Fig.~\ref{fig:dirGraphExample}(b)-(d). Clearly, the Laplacian for directed graph is not symmetric, nevertheless, it follows some important properties: (i) sum of each row is zero and hence, $\lambda = 0$ is surely an eigenvalue, and (ii) real parts of the eigenvalues are non-negative for a graph with positive edge-weights.\\
\begin{figure}[h]
\centering
\begin{subfigure}[t]{0.24\textwidth}
\centering	
\includegraphics[scale=0.6]{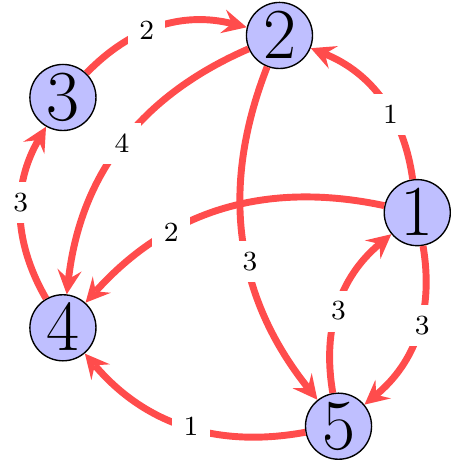}
\caption{An example directed graph.}
\end{subfigure}
\begin{subfigure}[t]{0.24\textwidth}
\centering
\footnotesize
\vspace{-2.2cm}
$\mathrm{W}$ = 
$ \begin{bmatrix}
        		0 & 0 & 0 & 0 & 3 \\
    			1 & 0 & 2 & 0 & 0 \\
    			0 & 0 & 0 & 3 & 0 \\
    			2 & 4 & 0 & 0 & 1 \\
    			3 & 3 & 0 & 0 & 0
  				\end{bmatrix}$
\caption{Weight matrix.}
\end{subfigure}\\
\vspace{0.4cm}
\begin{subfigure}[t]{0.24\textwidth}
\centering
\footnotesize
$\mathrm{D_{in}}$ = 
$ \begin{bmatrix}
        		3 & 0 & 0 & 0 & 0 \\
    			0 & 3 & 0 & 0 & 0 \\
    			0 & 0 & 3 & 0 & 0 \\
    			0 & 0 & 0 & 7 & 0 \\
    			0 & 0 & 0 & 0 & 6
  				\end{bmatrix}$
\caption{In-degree matrix.}
\end{subfigure}
\begin{subfigure}[t]{0.24\textwidth}
\centering
\footnotesize
$\mathrm{L}$ = 
$ \begin{bmatrix}
        		3 & 0 & 0 & 0 & -3 \\
    			-1 & 3 & -2 & 0 & 0 \\
    			0 & 0 & 3 & -3 & 0 \\
    			-2 & -4 & 0 & 7 & -1 \\
    			-3 & -3 & 0 & 0 & 6
    			\end{bmatrix}$
\caption{Laplacian matrix.}
\end{subfigure}
\caption{A directed graph and the corresponding matrices.}
\label{fig:dirGraphExample}
\end{figure}
\indent There also exist a few other definitions of graph Laplacian (normalized as well as combinatorial) for directed graphs~\cite{Li2012,Chung2005}. However, we choose Eq.~(\ref{eq:outLap}) as the definition of graph Laplacian for our analysis.
\subsection{Shift Operator}
\label{subsec:ShiftOp}
We identify the shift operator from the graph structure corresponding to a discrete time periodic signal and then extend it to arbitrary graphs. \\
\indent A finite duration periodic discrete signal can be thought of as a graph signal lying on a directed cyclic graph shown in Fig.~\ref{fig:dirCyclic}. Indeed, Fig.~\ref{fig:dirCyclic} is the support of a periodic time-series having a period of five samples. The directed Laplacian of the graph is
\begin{figure}[h]
\centering
\includegraphics[scale=0.5]{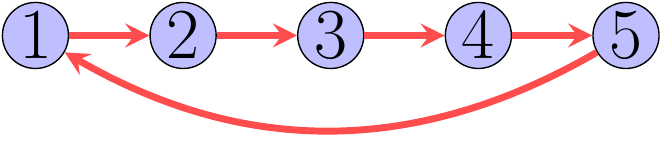}
\caption{A directed cyclic (ring) graph of five nodes. This graph is the underlying structure of a finite duration periodic (with period five) discrete signal. All edges have unit weights.}
\label{fig:dirCyclic}
\end{figure}
\begin{equation}
\label{eq:LapClassical}
\mathbf{L} = \begin{bmatrix}
        		1 & 0 & 0 & 0 & -1 \\
    			-1 & 1 & 0 & 0 & 0 \\
    			0 & -1 & 1 & 0 & 0 \\
    			0 & 0 & -1 & 1 & 0 \\
    			0 & 0 & 0 & -1 & 1
  				\end{bmatrix}.
\end{equation}
\noindent Consider a discrete time finite duration periodic signal $\mathbf{x}~=~[9~7~1~0~6]^T$ defined on the graph shown in Fig.~\ref{fig:dirCyclic}. Shifting the signal $\mathbf{x}$ by one unit to the right results in the signal $\mathbf{\tilde{x}} = [6~9~7~1~0]^T$. This shifted version of the signal can also be found as 
\begin{equation*}
\mathbf{\tilde{x}} = \mathbf{S}\mathbf{x} = (\mathrm{I} - \mathbf{L})\mathbf{x} = \begin{bmatrix}
        		0 & 0 & 0 & 0 & 1 \\
    			1 & 0 & 0 & 0 & 0 \\
    			0 & 1 & 0 & 0 & 0 \\
    			0 & 0 & 1 & 0 & 0 \\
    			0 & 0 & 0 & 1 & 0
  				\end{bmatrix} 
  				\begin{bmatrix}
        		9 \\
    			7 \\
    			1\\
    			0\\
    			6
  				\end{bmatrix} = \begin{bmatrix}
        		6 \\
    			9 \\
    			7\\
    			1\\
    			0
  				\end{bmatrix},
\end{equation*}
\noindent where $\mathbf{S} = (\mathrm{I} -\mathbf{L})$ is treated as the shift operator (matrix). \\
\indent We extend the notion of shift to arbitrary graphs and consequently use $\mathbf{S} = (\mathrm{I} -\mathbf{L})$ as the shift operator for graph signals. Hence, the shifted version of a graph signal $\mathbf{f}$ can be calculated as
\begin{equation}
\label{eq:shift}
\mathbf{\tilde{f}} = \mathbf{S}\mathbf{f} = (\mathrm{I} -\mathbf{L}) \mathbf{f}.
\end{equation} 
\indent In contrast to the use of the weight matrix as the shift operator, we opt for the pre-mentioned shift operator which involves the directed Laplacian. Our selection of this shift matrix results in a better and simpler frequency analysis which will become evident in the subsequent sections. 
\subsection{Total Variation}
\label{subsec:TV}
Total variation (TV) of a graph signal is a measure of total amplitude oscillations in the signal values with respect to the graph. As discussed in~\cite{Moura2014}, TV of a graph signal $\mathbf{f}$ with respect to the graph $\cal{G}$ can be given as
\begin{equation}
\label{eq:TV1}
\mathrm{TV}_{\cal{G}}(\mathbf{f}) = \sum_{i=1}^{N} \lvert \nabla_i(\mathbf{f}) \rvert,
\end{equation}
\noindent where $\nabla_i(\mathbf{f})$ is the derivative of the graph signal $\mathbf{f}$ at vertex $i$ and is defined as the difference between the values of the original graph signal $\mathbf{f}$ and its shifted version at vertex~$i$:
\begin{equation}
\label{eq:graphGrad}
\nabla_i(\mathbf{f}) = (\mathbf{f} - \mathbf{\tilde{f}} ) (i).
\end{equation}
\noindent From Eq.~(\ref{eq:shift}),~(\ref{eq:TV1}),~and~(\ref{eq:graphGrad}), we have
\begin{equation}
\label{eq:graphTV}
\mathrm{TV}_{\cal{G}}(\mathbf{f}) = \sum_{i=1}^{N} \rvert  f(i) - \tilde{f}(i)  \lvert 
= \lvert \lvert \mathbf{f} - \tilde{\mathbf{f}}  \rvert \rvert_1 
= \lvert \lvert  \mathbf{Lf} \rvert \rvert_1.
\end{equation}

\indent Observe that the quantity $\mathbf{Lf}$ at node $i$ is the sum of the weighted differences (weighted by corresponding edge weights) between value of $\mathbf{f}$ at node $i$ and values at the neighboring nodes. In other words, $(\mathbf{Lf})(i) =  \nabla_i(\mathbf{f})$ provides a measure of variations in the signal values as we move from node~$i$ to its adjacent nodes. Thus,~$\ell_1$-norm of the quantity~$\mathbf{Lf}$ can be interpreted as the absolute sum of the \textit{local variations} in~$\mathbf{f}$. We will utilize TV given by Eq.~(\ref{eq:graphTV}) to estimate variations of the graph Fourier basis and subsequently to identify low and high frequencies of a graph.
\subsection{Graph Fourier Transform based on Directed Laplacian}
\label{subsec:GFTDL}
DSP$_\mathrm{G}$ derives analogy from the perspective of LSI filtering. In DSP$_\mathrm{G}$, LSI graph filters are polynomials in the graph weight matrix $\mathrm{W}$ and as a result the eigenvectors of $\mathrm{W}$ become the eigenfunctions of LSI graph filters. Analogous to the fact that the complex exponentials are invariant to LSI filtering in classical signal processing, the eigenvectors of $\mathrm{W}$ are utilized as the graph Fourier basis. \\
\indent We redefine the Graph Fourier Transform under DSP$_\mathrm{G}$. We consider the Jordan eigenvectors of the graph Laplacian matrix as the graph Fourier basis that are invariant to an LSI graph filter (discussed in Section~\ref{sec:Filtering}), which is a polynomial in~$\mathbf{L}$ under proposed shift operator. Using Jordan decomposition, the graph Laplacian is decomposed as
\begin{equation}
\label{eq:Lap_jordan}
\mathbf{L} = \mathbf{VJV}^{-1},
\end{equation}
\noindent where $\mathbf{J}$, known as the Jordan matrix, is a block diagonal matrix similar to $\mathbf{L}$ and the Jordan eigenvectors of $\mathbf{L}$ constitute the columns of $\mathbf{V}$. We define GFT of a graph signal $\mathbf{f}$ as
\begin{equation}
\label{eq:GFTDL}
\mathbf{\hat{f}} = \mathbf{V}^{-1}\mathbf{f}.
\end{equation}

\noindent Here, $\mathbf{V}$ is treated as the graph Fourier matrix whose columns constitute the graph Fourier basis. Inverse Graph Fourier Transform can be calculated as
\begin{equation}
\label{eq:GFTDLinverse}
\mathbf{f} = \mathbf{V}\mathbf{\hat{f}}.
\end{equation}
\indent In this definition of GFT, the eigenvalues of the graph Laplacian act as the graph frequencies and the corresponding Jordan eigenvectors act as the graph harmonics. The eigenvalues with small absolute value correspond to low frequencies and vice-versa; we will prove that shortly. Thus, the frequency order turns out to be \textit{natural}. \\
\indent Before discussing ordering of frequency, we consider a special case when the Laplacian matrix is diagonalizable.
\subsubsection{Diagonalizable Laplacian Matrix}
\label{subsubsec:DiagonalizableLap}
When the graph Laplacian is diagonalizable, Eq.~(\ref{eq:Lap_jordan}) is reduced to:
\begin{equation}
\label{eq:DiagonalizableGFTDL}
\mathbf{L = V \Lambda V}^{-1}.
\end{equation}
Here, $\Lambda \in \mathbb{C}^{N\times N}$ is a diagonal matrix containing the eigenvalues $\lambda_0, \lambda_1,\ldots,\lambda_{N-1}$ of $\mathbf{L}$ and $V=[\mathbf{v_0,v_1,\ldots,v_{N-1}}] \in \mathbb{C}^{N\times N}$ is the matrix with columns as the corresponding eigenvectors of $\mathbf{L}$. Note that for a graph with real non-negative edge weights, the graph spectrum will lie in the right half of the complex frequency plane (including the imaginary axis).
\subsubsection{Undirected Graphs}
\label{subsubsec:UndirCase}
For an undirected graph with real weights, the graph Laplacian matrix $\mathbf{L}$ is real and symmetric. As a result, the eigenvalues of $\mathbf{L}$ turn out to be real and $\mathbf{L}$ constitutes orthonormal set of eigenvectors. Hence, the Jordan form of the Laplacian matrix for undirected graphs can be written as
\begin{equation}
\mathbf{L = V\Lambda V}^T,
\end{equation}
where $\mathbf{V}^T = \mathbf{V}^{-1}$, because the eigenvectors of $\mathbf{L}$ are orthogonal in undirected case. Consequently, GFT of a signal $\mathbf{f}$ can be given as $\mathbf{\hat{f}} = \mathbf{V}^T\mathbf{f}$ and the inverse can be calculated as $\mathbf{f} = \mathbf{V}\mathbf{\hat{f}}$. Note that the graph spectrum will lie on the real axis of the complex frequency plane, given the weight matrix is real. Moreover, if the weights are non-negative as well, the graph spectrum will lie on the non-negative half of the real axis. This coincides with the GFT presented in~\cite{Wavelet2011,Shuman2013}, where only undirected graphs with real and non-negative weights were considered. Thus, the new definition of GFT unifies DSP$_\mathrm{G}$ and Laplacian based approach.
\subsubsection{Frequency Ordering}
\label{subsubsec:FrOrdering}
We utilize the definition of TV given by Eq.~(\ref{eq:graphTV}) to quantify oscillations in the graph harmonics, and subsequently to order the frequencies. The eigenvalues for which the corresponding proper eigenvectors have small variations are labeled as low frequencies and vice-versa. The frequency ordering is established by Theorem~\ref{th:frequency_ordering}. 
\begin{theorem}
\label{th:frequency_ordering}
Let $\mathbf{v}_i$ and $\mathbf{v}_j$ be the eigenvectors with unit $\ell_1$-norm corresponding to two distinct complex eigenvalues $\lambda_i,\lambda_j \in \mathbb{C}$ of the Laplacian matrix $\mathbf{L}$ of the graph $\cal{G}$. If $\lvert \lambda_i \rvert > \lvert \lambda_j \rvert $, then the TVs of these eigenvectors with respect to the graph $\cal{G}$ satisfy
\begin{equation}
\label{eq:frequency_ordering}
\mathrm{TV}_{\cal{G}}(\mathbf{v}_i) > \mathrm{TV}_{\cal{G}}(\mathbf{v}_j).
\end{equation}
\end{theorem}
\begin{proof}
Let $\mathbf{v}_r$ be the proper eigenvector corresponding to eigenvalue $\lambda_r$ of the Laplacian matrix $\mathbf{L}$ of the graph $\cal{G}$, then $\mathbf{Lv}_r = \lambda_r \mathbf{v}_r$. Now, using Eq.~(\ref{eq:graphTV}), TV of $\mathbf{v}_r$ with respect to the graph $\cal{G}$ is calculated as 
\begin{equation*}
\mathrm{TV}_{\cal{G}}(\mathbf{v}_r) = \lvert \lvert  \mathbf{Lv}_r  \rvert \rvert_1
							=   \lvert \lvert  \lambda_r \mathbf{v}_r  \rvert \rvert_1
							=  \lvert \lambda_r \rvert (\lvert \lvert  \mathbf{v}_r  \rvert \rvert_1).
\end{equation*}
If we scale all eigenvectors to have the same $\ell_1$-norm, then from the above analysis we note that 
\begin{equation}
\mathrm{TV}_{\cal{G}}(\mathbf{v}_r) \propto \lvert \lambda_r \rvert.
\end{equation}
\end{proof}
\begin{figure}[h]
\centering
  \begin{subfigure}[t]{0.26\textwidth}
				\centering
                \includegraphics[scale=0.36]{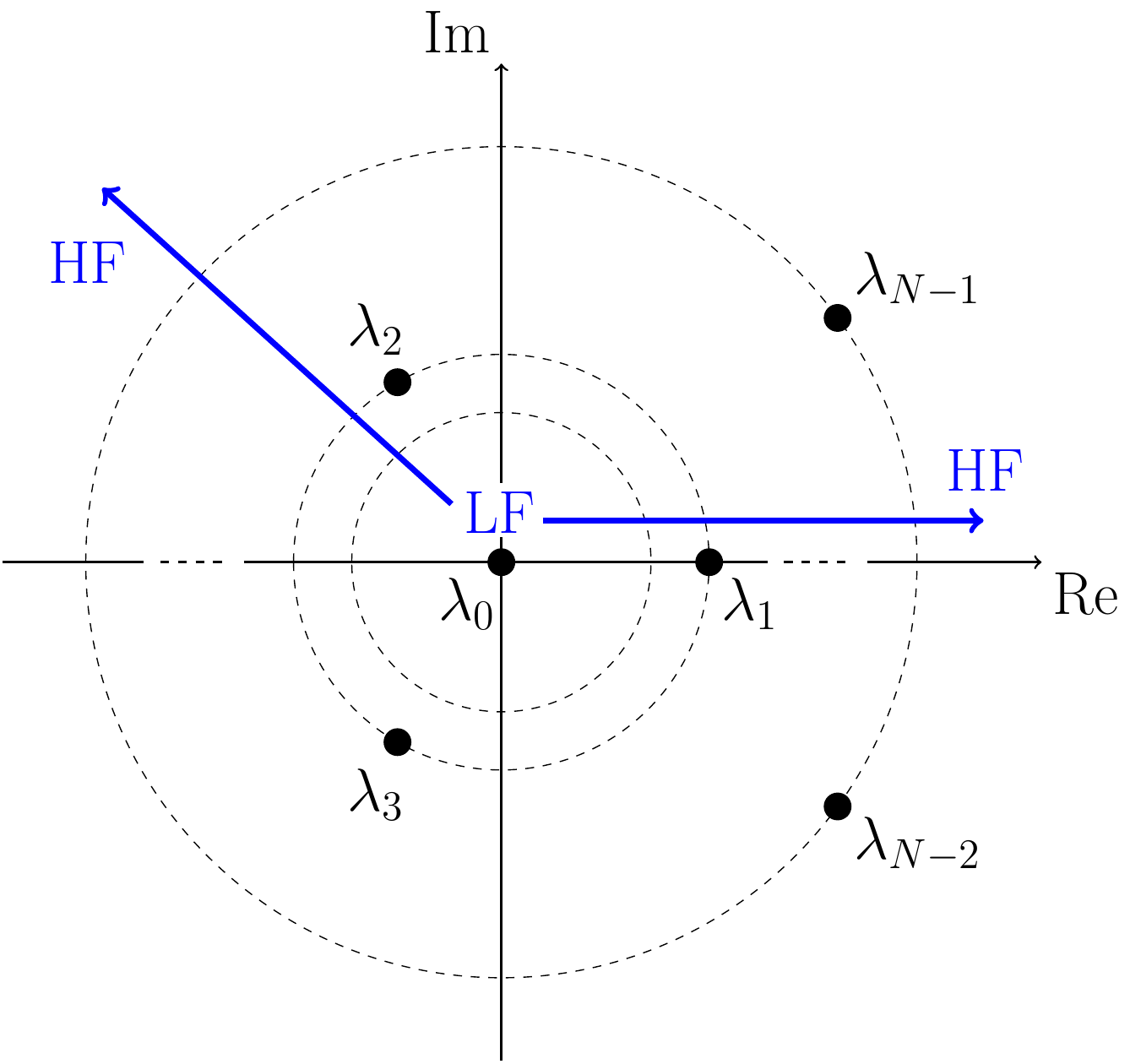}
                \caption{Ordering for a graph with positive or negative edge weights.}
        \end{subfigure} \hspace*{0.6cm}
        \begin{subfigure}[t]{0.18\textwidth}
				\centering
                \includegraphics[scale=0.36]{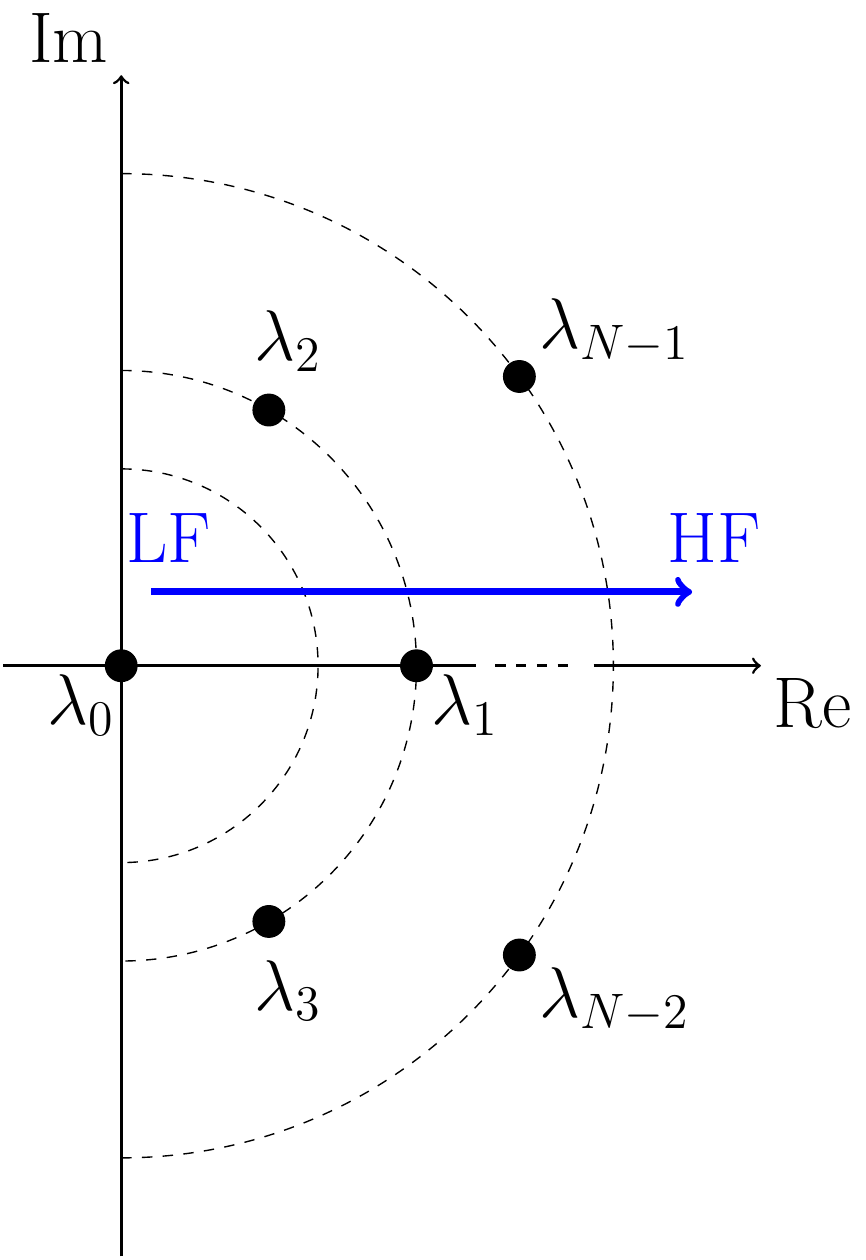}
				\caption{Ordering for a graph with positive edge weights.}        
        \end{subfigure}\\        
        \begin{subfigure}[t]{0.26\textwidth}
				\centering
                \includegraphics[scale=0.36]{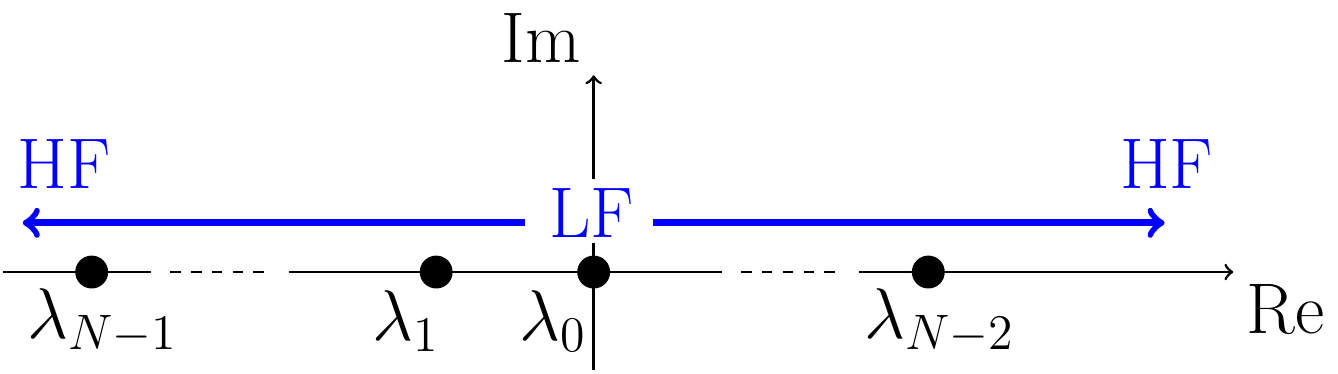}
                \caption{Ordering for an undirected graph with real edge weights.}
        \end{subfigure} \hspace*{0.6cm}
        \begin{subfigure}[t]{0.18\textwidth}
				\centering
                \includegraphics[scale=0.36]{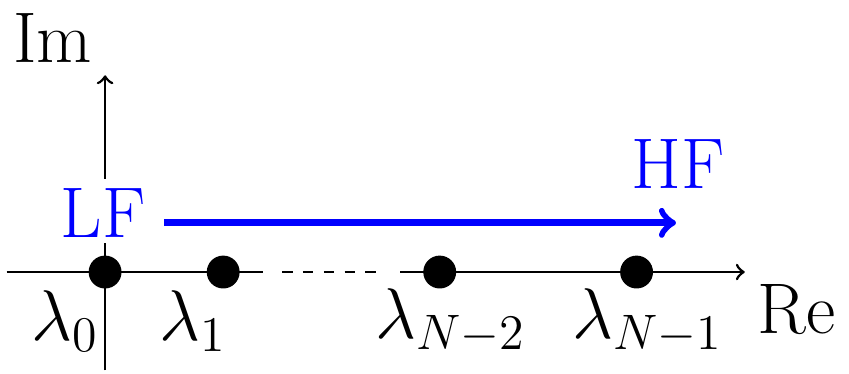}
				\caption{Ordering for an undirected graph with real and non-negative edge weights.}        
        \end{subfigure}     
\caption{Frequency ordering from low frequencies (LF) to high frequencies
(HF). As we move away from the origin (zero frequency) in the complex frequency plane,  the eigenvalues correspond to higher frequencies because the TVs of the corresponding eigenvectors increase.}
\label{fig:freqOrder}
\end{figure}
\indent Note that TV of a proper eigenvector is directly proportional to the absolute value of the corresponding eigenvalue. Therefore, all the proper eigenvectors corresponding to the eigenvalues having equal absolute value will have the same TV. As a result, distinct eigenvalues may sometimes yield exactly same TV. Because of this reason, sometimes frequency ordering is not unique. However, if all the eigenvalues are real, the frequency ordering is guaranteed to be unique. Fig.~\ref{fig:freqOrder} shows the visualization of frequency ordering in the complex frequency plane; $\lambda_0 = 0$ corresponds to zero frequency and as we move away from the origin, the eigenvalues correspond to higher frequencies. \\
\indent An important point worth mentioning here is that we achieve same order of frequency as given by Theorem~\ref{th:frequency_ordering} even if we use quadratic form (2-Dirichlet form) in place of TV. 2-Dirichlet form~\cite{Moura2014} of a signal $\mathbf{f}$ is defined as
\begin{equation}
\label{eq:2Dirichlet}
S_2(\mathbf{f})=\frac{1}{2}\sum_{i=1}^{N} \lvert \nabla_i(\mathbf{f}) \rvert ^2.
\end{equation}
Substituting $\nabla_i(\mathbf{f})$ from Eq.~(\ref{eq:graphGrad}), we have
\begin{equation}
\label{eq:GraphQuadraticForm}
S_2(\mathbf{f})  = \frac{1}{2} \sum_{i=1}^{N} \lvert  f(i) - \tilde{f}(i)  \rvert^2 
 = \frac{1}{2} \lvert \lvert \mathbf{f} - \tilde{\mathbf{f}}  \rvert \rvert_2^2 
 = \frac{1}{2} \lvert \lvert  \mathbf{Lf} \rvert \rvert_2^2.
\end{equation}
Now, quadratic form of a proper eigenvector $\mathbf{v}_r$ corresponding to frequency $\lambda_r$ can be calculated as $S_2(\mathbf{v}_r) = \frac{1}{2}\lvert \lvert  \mathbf{Lv}_r \rvert \rvert_2^2 = \frac{1}{2}\lvert \lvert \lambda_r \mathbf{v}_r \rvert \rvert_2^2=\frac{1}{2} \lvert \lambda_r \rvert^2  \lvert \lvert  \mathbf{v}_r \rvert \rvert_2^2$. Therefore, if all the eigenvectors are scaled to have same $\ell_2$-norm, then
\begin{equation}
\mathrm{S}_2(\mathbf{v}_r) \propto \lvert \lambda_r \rvert^2.
\end{equation}
\indent Therefore, the frequency ordering based on the quadratic form will be same as given by Theorem~\ref{th:frequency_ordering}. 
\subsubsection{Example}
\label{subsubsec:exampleGFT}
Let us consider a directed graph shown in Fig.~\ref{fig:dirGraphExample}(a). Performing Jordan decomposition of the Laplacian matrix, we find the Fourier matrix $\mathbf{V}$ and Jordan matrix $\mathbf{J}$ as follows.
\begin{equation*}
\scriptsize
\mathbf{V} =  \begin{bmatrix}
        		0.447 & 0.680 & -0.232 - 0.134i & -0.232 + 0.134i& -0.535 \\
    			0.447 & -0.502 & 0.232 + 0.312i &0.232 - 0.312i & 0.080 \\
    			0.447 & -0.502 & -0.502 - 0.201i & -0.502 + 0.201i& 0.080 \\
    			0.447 & -0.108 & 0.618 - 0.089i & 0.618 + 0.089i & -0.125 \\
    			0.447 & 0.146 & 0.309 & 0.309 & 0.828
  				\end{bmatrix}
\end{equation*}
\begin{equation*}
\small
\mathbf{J} =  \mathrm{diag}\{0,~ 2.354,~6.000 - 1.732i,~ 6.000 + 1.732i,~7.646\}.
\end{equation*}
\noindent Eigenvalue~$\lambda~=~7.646$ corresponds to the highest frequency of the graph. Also note that the TVs of the eigenvectors corresponding to the frequencies $\lambda = 6 - 1.732i$ and $\lambda = 6 + 1.732i$ are equal because both of the frequencies have same absolute value. The magnitudes of the GFT coefficients of the graph signal~$\mathbf{f}=[0.12~0.38~0.81~0.24~0.88]^T$ are plotted in Fig.~\ref{fig:dirGraphspectrum}.
\begin{figure}[h]
\centering
\includegraphics[scale=0.23]{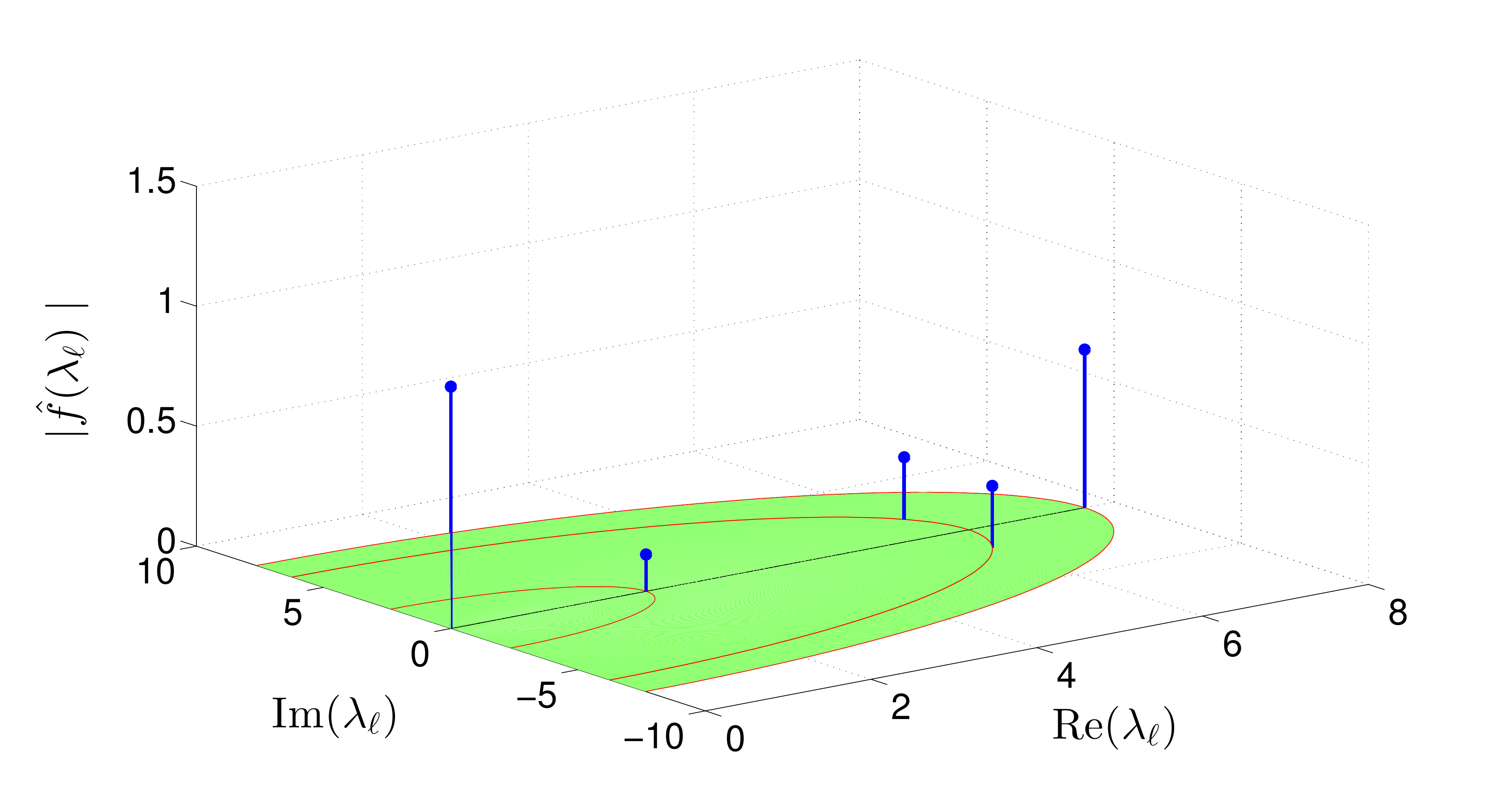}
\caption{Magnitude spectrum of the graph signal~$\mathbf{f}=[0.12~0.38~0.81~0.24~0.88]^T$ defined on the graph shown in Fig.~\ref{fig:dirGraphExample}(a).}
\label{fig:dirGraphspectrum}
\end{figure}
\subsubsection{Concept of Zero Frequency}
The Jordan eigenvector of $\mathbf{L}$ corresponding to the zero eigenvalue is given as~$\mathbf{v}_0~=~\frac{1}{\sqrt{N}}[1~1\ldots~1]^T$. Therefore, for a constant graph signal, GFT will have only a single non-zero coefficient at zero frequency (eigenvalue). For example, consider a constant graph signal $\mathbf{f} = [1~1~1~1~1]^T$ residing on the graph shown in Fig.~\ref{fig:dirGraphExample}(a). GFT of the signal is given by $\mathbf{\hat{f}} = [\sqrt{5}~0~0~0~0]^T$, magnitudes of which are plotted in Fig.~\ref{fig:dirGraphspectrumconst}.
\begin{figure}[h]
\centering
\includegraphics[scale=0.23]{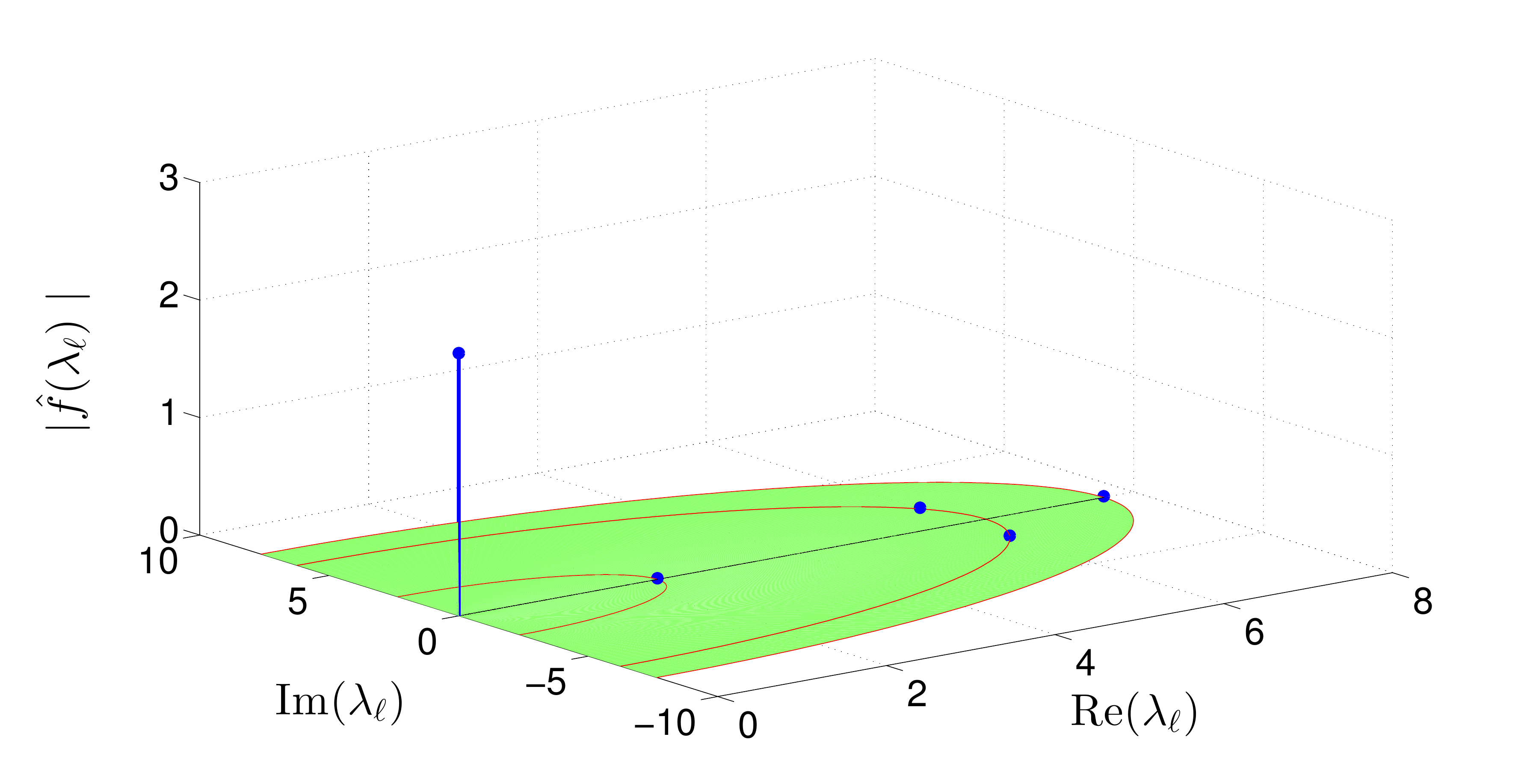}
\caption{Magnitude spectrum of a constant signal $f = [1~1~1~1~1]^T$ defined on the graph shown in Fig.~\ref{fig:dirGraphExample}(a).}
\label{fig:dirGraphspectrumconst}
\end{figure}

\indent We can observe the presence of only zero frequency component in the spectrum of a constant graph signal. This is ``an evidence'' to the intuition that the variation in a constant graph signal is zero as we travel from a node to the other node connected by a directed edge. In contrast, the graph Fourier transform defined in DSP$_\mathrm{G}$ fails to give this basic intuition. 
\section{Linear Shift Invariant Graph Filters}
\label{sec:Filtering}
An LSI graph filter is a linear filter for which shifted version of the filter output is same as the filter output to the shifted input. That is, if $\mathbf{S}(\mathbf{f}_{out}) = \mathbf{H}(\mathbf{S}\mathbf{f}_{in})$, where $\mathbf{f}_{out}$ is the output of the filter $\mathbf{H}$ corresponding to input $\mathbf{f}_{in}$, then the filter $\mathbf{H}$ is shift-invariant (SI). The condition for a graph filter to be LSI is given by Theorem~\ref{th:filterLSI}.
\begin{theorem}
\label{th:filterLSI}
A graph filter $\mathbf{H}$ is LSI if the following conditions are satisfied:
\begin{enumerate}
\item Geometric multiplicity of each distinct eigenvalue of the graph Laplacian is one.
\item The graph filter $\mathbf{H}$ is a polynomial in $\mathbf{L}$, i.e., $\mathbf{H}$ can be written as
\end{enumerate}
\begin{equation}
\label{eq:LSI_filter}
\mathbf{H} = h(\mathbf{L}) = \sum_{m=0}^{M-1}h_m\mathbf{L}^{m} = h_0 \mathbf{I} + h_1 \mathbf{L} + \ldots +h_{M-1} \mathbf{L}^{M-1},
\end{equation}
\indent where $h_0,h_1,~\ldots~,h_{M-1} \in \mathbb{C}$ are called filter taps.
\end{theorem}
\begin{proof}
From Eq.~(\ref{eq:shift}), the shift operator is $\mathrm{S = I }- \mathbf{L}$. Then, for a filter $\mathbf{H}$ to be LSI, the following must be satisfied:
\begin{align*}
\mathbf{S}(\mathbf{Hf}) &= \mathbf{H}(\mathbf{S}\mathbf{f})\\
\mathrm{or,}\quad\quad (\mathrm{I} - \mathbf{L})(\mathbf{Hf}) &= \mathbf{H}((\mathrm{I} - \mathbf{L}) \mathbf{f})\\
\mathrm{or,}~~\quad\quad\quad\quad\quad\mathbf{LH} &= \mathbf{HL}.
\end{align*}
In other words, the filter is LSI if the matrices $\mathbf{L}$ and $\mathbf{H}$ commute. This is true when $\mathbf{H}$ is a polynomial in $\mathbf{L}$, given that the characteristic polynomial and the minimal polynomial of $\mathbf{L}$ are equal or the geometric multiplicity of every distinct eigenvalue is one~\cite{Moura2013}.
\end{proof}
Considering a $2$-tap graph filter and substituting $h_0=1$ and $h_1 = -1$ in Eq.~(\ref{eq:LSI_filter}), we get $\mathbf{H=S} = \mathrm{I}-\mathbf{L}$, which is the shift operator discussed in Section~\ref{subsec:ShiftOp}. Thus, the shift operator is a first order LSI filter. Also, the graph Fourier basis, which are the Jordan eigenvectors of $\mathbf{L}$, are the eigenfunctions of an LSI filter described by Eq.~(\ref{eq:LSI_filter}). 
\section{Conclusions}
\label{sec:Conclusion}
In this paper, under the DSP$_\mathrm{G}$ framework, we considered the Jordan eigenvectors of the directed Laplacian as graph harmonics and the corresponding eigenvalues as the graph frequencies, and utilized that to redefine GFT. For this purpose, we proposed a shift operator derived from the directed Laplacian of a graph. Our proposed definition of GFT under DSP$_\mathrm{G}$ framework, provides natural frequency interpretation and links DSP$_\mathrm{G}$ to the Laplacian based approach. We also showed that by considering our proposed shift operator, LSI filters become polynomial in the directed Laplacian.


\bibliographystyle{IEEEtran}

\end{document}